\documentclass[letterpaper, 10 pt, conference]{ieeeconf}  % Comment this line out if you need a4paper
\IEEEoverridecommandlockouts                              % This command is only needed if 
\usepackage{amsmath,amsfonts, xcolor}
\usepackage{graphicx}
\usepackage{amssymb}
\usepackage{bm}
\usepackage{balance}

\usepackage{amsthm}
\newtheorem{thm}{Theorem}

\newtheorem{lem}{Lemma}

\theoremstyle{remark}
\newtheorem{defn}{Definition}

\newtheorem{remark}{Remark}

\newcommand{\R}{\mathbb{R}}
\newcommand{\G}{\bm{G}}

\title{\LARGE \bf
Robustly Invertible Nonlinear Dynamics and the BiLipREN: Contracting Neural Models with Contracting Inverses
}

\author{Yurui Zhang, Ruigang Wang and Ian R. Manchester% <-this % stops a space
\thanks{This work was supported in part by the Australian Research Council.}% <-this % stops a space
\thanks{The authors are with the Australian Centre for Robotics (ACFR), and the School of Aerospace, Mechanical and Mechatronic Engineering, The University of Sydney, Sydney, NSW 2006, Australia {\tt\small ian.manchester@sydney.edu.au}}%
}

\begin{document}

\maketitle
\thispagestyle{empty}
\pagestyle{empty}

%%%%%%%%%%%%%%%%%%%%%%%%%%%%%%%%%%%%%%%%%%%%%%%%%%%%%%%%%%%%%%%%%%%%%%%%%%%%%%%%

\begin{abstract}
We study the invertibility of nonlinear dynamical systems from the perspective of contraction and incremental stability analysis and propose a new invertible recurrent neural model: the BiLipREN. In particular, we consider a nonlinear state space model to be robustly invertible if an inverse exists with a state space realisation, and both the forward model and its inverse are contracting, i.e. incrementally exponentially stable, and Lipschitz, i.e. have bounded incremental gain. This property of bi-Lipschitzness implies both robustness in the sense of sensitivity to input perturbations, as well as robust distinguishability of different inputs from their corresponding outputs, i.e. the inverse model robustly reconstructs the input sequence despite small perturbations to the initial conditions and measured output. Building on this foundation, we propose a parameterization of neural dynamic models: bi-Lipschitz recurrent equilibrium networks (biLipREN), which are robustly invertible by construction. Moreover, biLipRENs can be composed with orthogonal linear systems to construct more general bi-Lipschitz dynamic models, e.g., a nonlinear analogue of minimum-phase/all-pass (inner/outer) factorization. We illustrate the utility of our proposed approach with numerical examples.

\end{abstract}

%%%%%%%%%%%%%%%%%%%%%%%%%%%%%%%%%%%%%%%%%%%%%%%%%%%%%%%%%%%%%%%%%%%%%%%%%%%%%%%%
\section{Introduction}

The study of dynamic system invertibility plays a fundamental role in control theory and applications. In the absence of any uncertainty, an ideal inverse of an input-output map provides a perfect feedforward controller, while in the feedback setting robust  control was famously characterised by Zames in terms of existence of an approximate inverse \cite{zames1981feedback}. Many practical techniques in control design and related areas are based in some way on dynamic system inversion. System inverses are used for feedforward control in precision motion control (e.g. \cite{markusson2001iterative, rigney2009nonminimum, van2018inversion, bolderman2024physics}) and for digital predistortion compensation in electronics (e.g. \cite{morgan2006generalized, ghannouchi2009behavioral, liu2006augmented, tanovic2018equivalent}). System inversion is used a component in feedback schemes such as nonlinear dynamic inversion for flight control (e.g. \cite{reiner1996flight, miller2011nonlinear, sieberling2010robust, wang2017robust}) and internal model control schemes which have been widely applied in the process industries (e.g. \cite{garcia1982internal, economou1986internal}). Other works such as \cite{celani2010output, wang2017robust} developed robust output feedback regulators for certain classes of invertible nonlinear MIMO systems.

In this work, we propose a definition of \emph{robust invertibility} for a nonlinear dynamic system, which roughly speaking means that a causal inverse exists and both the system and its inverse have guaranteed internal stability and bounded input-output behavior. Internal stability is characterized in terms of \emph{contraction} \cite{lohmiller1998contraction, FB-CTDS}, a global stability notion for nonlinear system that ensures all trajectories converge to each other. The bounded input-output behavior for both forward and inverse system is characterized by \emph{bi-Lipschitz} properties, meaning small changes in input lead to small changes in output, and different outputs correspond to robustly distinguishable inputs. These properties guarantee the inverse system can robustly reconstructs the the input sequences from output sequences with the presence of disturbance and different initial states.

\subsection{Inverses of Nonlinear Dynamical Systems}
Research on system inverses originated with the inverse of a linear system. A continuous-time linear single-input single-output (SISO) system is said to be minimum-phase if all the zeros of its transfer function lie in the open left half of the complex plane \cite{moylan1977stable}. A minimum-phase system is stable if and only if all its poles and zeros are in the left half-plane. For nonlinear systems, the concept of minimum-phase is defined through the internal dynamics. The zero dynamics refer to the internal dynamics of the system when the input is chosen such that the output remains identically zero. A nonlinear system is minimum-phase if its zero dynamics are (globally) asymptotically stable \cite{byrnes1988local}. 

For non-minimum-phase systems the inverse is unstable, but factorizations based on the classical linear minimum-phase/all-pass a.k.a. inner-outer factorizations can be developed \cite{ball1992inner, van2018l2}, and approximate inversion can be achieved e.g. via finite horizon non-causal preview \cite{zou2007precision} or pseudo-inverse methods \cite{romagnoli2019general}. Such factorizations are useful in control schemes generalizing the classical Smith predictor and internal model control \cite{laughlin1987smith, van2018l2}.

In much of the literature, invertibility of a nonlinear system is defined under the assumption of a well-defined \emph{relative degree} \cite{hirschorn1979invertibility, devasia1996nonlinear} and asymptotically stable zero dynamics. Roughly speaking, the relative degree at a point is the number of times the output must be differentiated (or delayed in discrete time) before the input appears explicitly. Associated with this process is a \emph{normal form}, i.e. a coordinate transformation that introduces separates system into nonlinear internal dynamics and a linear input-output component construct via partial feedback linearization \cite{liberzon2002output}. However, the restriction to systems of a fixed vector relative degree and the requirement of a normal form representation can be limiting in many applications.

Furthermore, while many existing approaches to study stability of the inversion in terms of asymptotically stable zero dynamics or input-to-state stability (ISS), they do not explicitly account for robustness of the inverse in the sense of amplification of signal perturbations. To our knowledge, \cite{liberzon2004output, liberzon2002output} introduced the concept of output–input stability, i.e. inputs are bounded by functions of outputs, but this requires a normal form construction and does not imply incremental stability as studied in this paper. 

\subsection{Invertible neural networks}

In the machine learning literature, compositions of invertible static maps known as normalizing flows \cite{papamakarios2021normalizing}  have received significant attention for modeling complex probability distributions. A normalizing flow maps a complex data distributions to a simple distribution such as a Gaussian, which can then be inverted as a generative model for tasks motion generation \cite{luo2024potential}, imitation learning \cite{urain2020imitationflow}, time series forecasting \cite{de2020normalizing}, and nonlinear dynamics modeling \cite{drygala2022generative} in robotics. As more applications involve sequence-to-sequence mapping, dynamic invertible systems are well-suited for scenarios where the output depends on inputs over time, which is usually concerned in robotics. In the literature, the invertible residual layer $F(x) = x+H(x)$ is closely relate to our model, where the nonlinear block $H$ is a network with Lipschitz bound less than $1$ \cite{chen2019residual}. In \cite{behrmann2021understanding}, bi-Lipschitzness is introduce to mitigate the exploding inverse problem of invertible neural networks. In \cite{wangmonotone}, a strongly monotone  and Lipschitz residual layer is constructed with certified bound. A bi-Lipschitz network structure (biLipNet) is obtained based on the composition of these layers with orthogonal affine layers. 

The present paper extends the biLipNet concept to dynamical systems. In contrast to much of the literature on nonlinear system inversion, rather than construct an inverse of a given system, we propose to learn dynamical models which are robstly invertible by construction.

\subsection{Contributions}
    \begin{itemize}
        \item  We introduce a definition of robust invertibility of nonlinear dynamics from the perspective of contraction and bi-Lipschitzness, which the system can be robustly inverted with the existence of disturbance and difference in initial conditions. 
        \item We explicitly describe a bi-Lipschitz dynamical system using strong input-output monotonicity (a special case of bi-Lipschitzness). By composing such models in layers with static/dynamic linear orthogonal layers, we obtain more general bi-Lipschitz model representations. 
        \item We build on our previous work \cite{wangmonotone} and \cite{revay2024recurrent} proposing the bi-Lipschitz recurrent equilibrium networks (biLipREN) with guaranteed robust invertibility. The proposed model is effective in dynamic inversion, and generative models, which generalises the idea of minimum-phase linear systems, and are easy to control by inversion.
    \end{itemize}

{\bf Notation.} Let $\mathbb{N}$ and $\mathbb{R}$ be the set of natural and real numbers, respectively. We denote the set of sequences $x:\mathbb{N}\rightarrow \mathbb{R}^n$ by $\ell^n$. We use $\ell_2^n\subset \ell^n$ to denote the set of sequences with finite $\ell_2$ norm, i.e., $\|x\|=\sqrt{\sum_{t=0}^{\infty} |x_t|^2}<\infty$ where $|\cdot|$ is the Euclidean norm. We use $\|x\|_T:=\sqrt{\sum_{t=0}^{T}|x_t|^2}$ to denote the truncated norm of $x\in \ell^n$ over $T$ with $T\in \mathbb{N}$. 

\section{Preliminaries}

We consider the nonlinear state-space systems of the form
\begin{equation} \label{eq:system}
    x_{t+1}=f(x_t,u_t),\quad y_t=h(x_t,u_t)
\end{equation}
In the above, $x_t\in\R^n $, $ u_t, y_t\in \R^m$ are the model state, input and output respectively. 
If the output mapping $h_\theta$ is invertible w.r.t. its second argument, we can construct the causal state space realization of the inverse of \eqref{eq:system} as follows
\begin{equation}\label{eq:system-inv}
x_{t+1}=f(x_t,h^{-1}(x_t,y_t)),\quad u_t=h^{-1}(x_t,y_t)
\end{equation}
where $h^{-1}:\R^n\times \R^m\rightarrow\R^m$ satisfies
\begin{equation}
    h(x,h^{-1}(x,y))=y,\quad h^{-1}(x,h(x,u))=u,\; \forall x,u,y.
\end{equation}
A stable invertible system does not necessarily imply its inverse is stable. The particular form of stability we consider in this paper contraction:
\begin{defn}
    System \eqref{eq:system} is said to be \emph{contracting} if for two copies of the system with different initial conditions $a,b\in\mathbb{R}^{n}$, but the same input sequence $u\in \ell^m$, the corresponding state sequences $x^a,x^b$ satisfy 
    \begin{equation}\label{eq:contracting}
        |x_t^a-x_t^b|\leq \kappa \alpha^t |a-b|, \quad \forall t\in \mathbb{N}
    \end{equation}
    for some $\kappa>0$ and $\alpha \in [0,1)$.
\end{defn}

Let $\G:\ell^m\mapsto \ell^m$ be an operator from $u$ to $y$. To discuss causality we define the past projection operator
$\bm{P}_{t}$:
\begin{equation}
    (\bm{P}_{T} u)(\tau) = 
\begin{cases}
u_\tau, & \tau \leq T \\
0, & \tau > T
\end{cases}
\end{equation}
\begin{defn}
An operator $\G:\ell^m\mapsto\ell^n$ is causal if for any input sequence $u\in \ell^m$
\begin{equation}
    \bm{P}_T\G(\bm{P}_T u)=\bm{P}_T\G(u)\quad \forall T\in \mathbb{N}
\end{equation}
\end{defn}
\begin{defn}\label{defn:inv}
    An operator $\G:\ell^m\mapsto \ell^m$ is invertible if there exists an operator $\G^{-1}:\ell^m\mapsto \ell^m$  for any initial state $a\in \R^n$, such that 
    \begin{equation}\label{eq:inv-opt}
        \G^{-1}(\G(u))=u,\quad \G(\G^{-1}(y))=y
    \end{equation}
    for all $u,y\in \ell^m$.
\end{defn}

\begin{defn}
    A causal operator is said to be \emph{$\nu$-Lipschitz} with $\nu>0$ if
    \begin{equation}\label{eq:lipschitz}
        \|\G(u)-\G(v)\|_T\leq \nu \|u-v\|_T,\quad \forall T\in \mathbb{N}
    \end{equation}
    for all input sequences $u,v\in \ell^m$. Similarly, it is said to be \emph{$\mu$-inverse Lipschitz} with $\mu>0$ if
    \begin{equation}
        \|\G(u)-\G(v)\|_T\geq \mu \|u-v\|_T,\quad \forall T\in \mathbb{N}
    \end{equation}
    for all $a\in \R^n$ and $u,v\in \ell^m$. It is said to be \emph{bi-Lipschitz} with $\nu\geq \mu>0$, or simply $(\mu,\nu)$-biLipschitz, if it is $\mu$-inverse Lipschitz and $\nu$-Lipschitz. The ratio $\tfrac{\nu}{\mu}$ is referred to as a \textit{distortion bound}.
\end{defn}

\begin{remark}An invertible matrix is bi-Lipschitz with $\mu, \nu$ given by the smallest and largest singular values, respectively. The condition number corresponds to the distortion bound $\tfrac{\nu}{\mu}$.
\end{remark}

We frequently consider the operator induced by the system \eqref{eq:system} starting from initial conditions $x_0=a$, which we denote by $\G_a$.

\section{Robustly Invertible Dynamical Models}
In this section, we introduce the definition for \emph{robust invertibility} and explicitly derive the error bounds for a robustly invertible system in terms of contraction and bi-Lipschitzness. We introduce \textit{strong input-output monotonicity} as a special class of bi-Lipschitz systems. By composing strongly monotone models, we develop more general families bi-Lipschitz dynamic models.

\subsection{Robust Invertibility}
In many applications, it is useful to construct a model that allows us to robustly reconstruct the input sequence from the disturbed output, e.g. measurement noise, sensor noise or external disturbance. To be specific, let us consider the problem of recovering $u$ from $\tilde{y}=\G_a(u+\delta_u)$ where $a\in \R^n$ and $\delta_u\in \ell^m$ are the unknown initial state and input perturbation of the system \eqref{eq:system}, respectively. 
We formally define \emph{robust invertibility} in the following sense
\begin{defn}\label{def:r-invert}
    A system \eqref{eq:system} is said to be robustly invertible with $\alpha_{u}, \beta_{u},\alpha_{y},\beta_{y}>0$, if
    \begin{equation}\label{eq:robust-invert}
        \begin{split}
       &\|\G_b^{-1}(\G_a(u+\delta_u))-u\|_T\leq \alpha_{u}|a-b|+\beta_{u}\|\delta_u\|_T,\\
       &\|\G_a(\G_b^{-1}(y+\delta y))-y\|_T\leq\alpha_{y}|a-b|+\beta_{y}\|\delta_y\|_T
   \end{split}
    \end{equation}
   for any input sequences $u,\delta_u\in \ell^m$ and initial states $a,b\in \R^n$ and $\forall T\in \mathbb{N}$. $\G_a, \G_b^{-1}$ are the operators of \eqref{eq:system} and \eqref{eq:system-inv} with initial states $a,b$, respectively.
\end{defn}
We now relate robust invertibility to bi-Lipschitzness and contraction of a system operator $\G$ and its inverse $\G^{-1}$.
\begin{thm}
   Suppose that System \eqref{eq:system} is $(\mu,\nu)$-biLipschitz and contracting with rate $\alpha_1$ and overshoot $\kappa_1$, and its inverse is $(1/\nu,1/\mu)$-biLipschitz and contracting with rate $\alpha_2$ and overshoot $\kappa_2$. If the output mapping $h(x,u)$ is $(\gamma_1,\gamma_2)$-biLipschitz w.r.t. $x$, then for any input sequences $u,\delta_u\in \ell^m$, initial states $a,b\in \R^n$ and $\forall T\in \mathbb{N}$ we have 
   \begin{equation}\label{eq:error-bound}
   \begin{split}
       &\|\G_b^{-1}(\G_a(u+\delta_u)-u\|_T\leq \frac{\kappa_1\gamma_2}{\mu\sqrt{1-\alpha_1^2}}|a-b|+\frac{\nu}{\mu}\|\delta_u\|_T\\
       &\|\G_a(\G_b^{-1}(y+\delta_y)-y\|_T\leq \frac{\kappa_2\nu}{\gamma_1\sqrt{1-\alpha_2^2}}|a-b|+\frac{\nu}{\mu}\|\delta_y\|_T
   \end{split}
   \end{equation}
   i.e. System \eqref{eq:system} is robustly invertible.
\end{thm}
\begin{proof}
    First, applying triangle inequality to the first inequality we have
    \begin{equation}\label{eq:tri-1}
        \begin{split}
            \|\hat{u}-u\|_T&=\|\G_b^{-1}(\G_a(u+\delta_u))-u\|_T \\
            &\leq \|\G_b^{-1}(\G_a(u+\delta_u))-\G_b^{-1}(\G_a(u))\|\\
            &\quad+\|\G_b^{-1}(\G_a(u))-u\|.
        \end{split}
    \end{equation}
    where $\hat{u} = \G_b^{-1}(\G_a(u+\delta_u)-u$. Since \eqref{eq:system} is $(\mu,\nu)$-Lipschitz, then $\G_a$ is $(\mu,\nu)$-Lipschitz and $\G_b^{-1}$ is $(1/\nu, 1/\mu)$-Lipschitz for all $a,b$. This further implies 
    \begin{equation}\label{eq:tri-2}
        \begin{split}
            \|\G_b^{-1}(\G_a&(u+\delta_u))-\G_b^{-1}(\G_a(u))\| \\
            \leq & \frac{1}{\mu}\|\G_a(u+\delta u)-\G_a(u)\|_T\leq \frac{\nu}{\mu}\|\delta_u\|_T.
        \end{split}
    \end{equation}
    For the second term in \eqref{eq:tri-1}, we have
    \begin{equation*}
        \begin{split}
            \|\G_b^{-1}(\G_a(u))&-u\|_T
            =\|\G_b^{-1}(\G_a(u))-\G_b^{-1}(\G_b(u))\|_T\\
            \leq & \frac{1}{\mu}\|\G_a(u)-\G_b(u)\|_T \leq \frac{\gamma_2}{\mu}\left\|x^a-x^b\right\|_T
        \end{split}
    \end{equation*}
    where $x^a, x^b$ are the state trajectories of \eqref{eq:system} with the input $u$ and initial conditions $a,b$, respectively. Since \eqref{eq:system} is contracting with rate $\alpha$ and overshoot $\kappa$, we have
    \begin{equation*}
        \left\|x^a-x^b\right\|_T\leq \sqrt{\sum_{t=0}^T\kappa_1^2\alpha_1^{2t}|a-b|^2}\leq \frac{\kappa_1}{\sqrt{1-\alpha_1^2}}|a-b|,
    \end{equation*}
    which further implies that 
    \begin{equation}\label{eq:tri-3}
        \|\G_b^{-1}(\G_a(u))-u\|_T\leq \frac{\kappa_1\gamma_2}{\mu\sqrt{1-\alpha_1^2}}|a-b|.
    \end{equation}
    Finally, \eqref{eq:error-bound} follows by \eqref{eq:tri-1} - \eqref{eq:tri-3}.
    The second inequality follows the same logic.
    
\end{proof}

\subsection{Strong Input-output Monotonicity}
We introduce the strong input-output monotonicity which will be used to formulate the basic building block of bi-Lipschitz systems.
\begin{defn}\label{dfn:monotone}
    System \eqref{eq:system} is said to be ($\xi, \rho$)-\emph{strongly input-output monotone} with $\xi, \rho>0$ if for all $u,v\in\ell^m$ we have
    \begin{equation}\label{eq:in_mono}
       2 \langle \Delta y , \Delta u\rangle_T \geq \xi\|\Delta u\|^2_T+\rho\|\Delta y \|^2_T
    \end{equation}
    where $\Delta y =\G_a(u)-\G_a(v)$, $\Delta u = u-v$.
\end{defn}
We show that if an operator satisfies the strong input-output monotonicity, then it is bi-Lipschitz, as depicted in Fig.~\ref{fig:lipschitz}.
\begin{lem}\label{lem:iqc}
        If an operator $\G:\ell^m\mapsto\ell^m$ satisfies \eqref{eq:in_mono} with
        \begin{equation}\label{eq:monotone}
            \xi = \frac{2\mu\nu}{\mu+\nu}, \quad \rho = \frac{2}{\mu+\nu}
        \end{equation}
        then it is $(\mu,\nu)$-biLipschitz.
\end{lem}
\begin{proof}
For any pair of input sequences $u, v\in \ell_m$, we define $\Delta u = u-v$ and $\Delta y=y-z$ where $y=\G_a(u)$ and $z=\G_a(v)$. From \eqref{eq:in_mono}, we can obtain
\begin{equation*}
    2 \langle \Delta y , \Delta u\rangle_T\geq \frac{2}{\mu+\nu} (\mu\nu\|\Delta u_t\|^2+\|\Delta y_t\|^2), \;\forall T\in \mathbb{N}.
\end{equation*}
By Cauchy-Schwarz inequality, we have
\begin{equation*}
    (\mu+\nu)\|\Delta y\|_T\|\Delta u\|_T\geq \mu\nu \|\Delta u\|_T^2+\|\Delta y\|_T^2),
\end{equation*}
which further implies
\begin{equation}
    (\|\Delta y\|_T-\mu \|\Delta u\|_T)(\nu\|\Delta u\|_T-\|\Delta y\|_T)\geq 0.
\end{equation}
Thus, $\G$ is $(\mu,\nu)$-biLipschitz.
\end{proof}
\begin{figure}[!bt]
\centering
    \includegraphics[width=0.7\linewidth]{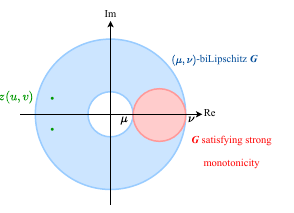}
    \caption{The Scaled Relative Graph \cite{chaffey2023graphical} is defined as a set of complex numbers with 
 $z(u, v) := \left\{ \frac{\| y - z \|}{\| u - v \|}
e^{\pm j \angle (u - v, y - z)}
\;\middle|\;
y = \G(u), z = \G(v) \right\}$, where $\angle (u, y) := \text{acos} {(\operatorname{Re} \langle u , y \rangle}{\|u\|^{-1} \|y\|^{-1})} \in [0, \pi]$, illustrated with green dots. This SRG depicts the incremental input-output properties of a $(\mu,\nu)$-biLipschitz operator $\G$. The ring (blue area) is for all $(\mu,\nu)$-biLipschitz operator $\G$ while the small circle (red area) is for the operator $\G$ satisfying the strong input-output monotonicity condition \eqref{eq:in_mono}.}\label{fig:lipschitz} 
\end{figure}

\subsection{Bi-Lipschitz Dynamic Models}
As shown in the previous section, strong input-output monotonicity is a special case of bi-Lipschitzness. However   monotonicity and bi-Lipschitzness differ in  their composition behavior. Given a $(\mu_1,\nu_1)$-biLipschitz operator $\G_1$, and $(\mu_2,\nu_2)$-biLipschitz operator $\G_2$, their composition $\G=\G_1\circ\G_2$ is also $(\mu_1\mu_2,\nu_1\nu_2)$-biLipschitz. However, the composition of two strongly monotone $\G_1, \G_2$ does not need to be strongly monotone, while it is still biLipschitz.

A static orthogonal layer \cite{trockman2021orthogonalizing} is an affine map realized by $\bm{O}_k=P_ku+q_k$ with $P_k^\top P_k =I$ and $q_k$ is the bias. $\bm{O}_k$ has an explicit inverse $\bm{O}_k^{-1} = P_k^\top(y-q_k)$ .

By composing strongly input-output monotone models with affine orthogonal layers, we can obtain more general biLipschitz dynamics 
\begin{equation}
    \bm{F}(u) = \bm{O}_{K+1}\circ\G_K\circ\bm{O}_{K}\circ\cdots\circ\bm{O}_{2}\circ\G_{1}\circ\bm{O}_1(u)
\end{equation}
where $\bm{O}_k$ is the $k\text{th}$ static orthogonal layer and $\G_k$ is the $k\text{th}$ strongly input-output monotone dynamic models. By the composition rule, the above general dynamic model is $(\mu,\nu )$-biLipschitz with $\mu = \Pi^K_{k=1}\mu_k$ and  $\nu = \Pi^K_{k=1}\nu_k$.  The composition of orthogonal layers and bi-Lipschitz models improves the model expressivity. The inverse of $\bm F$ can be calculated by passing backwards through the layers:
\begin{equation}
    \bm{F}^{-1}(y) = \bm{O}^{-1}_{1}\circ\G^{-1}_1\circ\cdots\circ\bm{O}^{-1}_{K}\circ\G^{-1}_{K+1}\circ\bm{O}^{-1}_K(y).
\end{equation}
\begin{remark}
    This kind of construction was studied for invertible static maps in \cite{wangmonotone}. As a basic example, any invertible matrix can be represented via its singular value decomposition $USV$, which is in the above form since $U$ and $V$ are orthogonal and $S$ is strongly input-output monotone, with $\mu, \nu$ the smallest and largest singular values.
\end{remark}

\subsection{Dynamic Orthogonal Layer with Anti-Causal Inverse}
A dynamic orthogonal layer $\bm{O}_k$ can be realized by an affine dynamic system in the form of \cite{heuberger2005modelling}
\begin{equation}\label{eq:dyn-orth}
    \bm{O}_k:\left\{\begin{aligned}
        h_{t+1} &= A_kh_t+B_ku_t+d_k\\
        y_t &=C_kh_t+D_ku_t+w_k
    \end{aligned}\right.
\end{equation}
with $h \in \mathbb{R}^p, u \in \mathbb{R}^m, y \in \mathbb{R}^m$, and $d_k\in\mathbb{R}^p,w_k\in\mathbb{R}^m$ are bias terms. $\bm{O}_k$ is orthogonal if and only if it has a state-space representation $A_k,B_k,C_k,D_k$ for which the block matrix
\begin{equation*}
Q_k=\begin{bmatrix}
        \begin{array}{c|c}
        A_k&B_k\\
        \hline
        C_k&D_k\\
    \end{array}
\end{bmatrix}
\end{equation*}
is orthogonal with $Q_k^\top Q_k=I$. The associated transfer function of System \eqref{eq:dyn-orth} is $H(z)= C_k(zI-A_k)^{-1}B_k+D_k$ with $H(z)H^\top(z^{-1})=I$, i.e. it is an \emph{all-pass filter}, see \cite{heuberger2005modelling} for further discussion.
 
The system $\bm{O}_k$ is a causal and stable, but its inverse of $\bm{O}_k$ in causal form is unstable. However $\bm{O}_k$ has an explicit anti-causal inverse:
\begin{equation}\label{eq:dyn-orth-inv}
      \bm{O}_k^{-1}:\left\{\begin{aligned}
        h_{t} &= A_k^\top( h_{t+1}-d_k)+C_k^\top (y_t-w_k)\\
        u_t &=B_k^\top (h_{t+1}-d_k)+D_k^\top (y_t-w_k)
    \end{aligned}\right.
\end{equation}
which is stable in reverse time \cite{heuberger2005modelling}. Hence this can be applied for robust inversion finite-length (batch) data.

\begin{remark}
    Both static and dynamic orthogonal layer can be parameterized by Cayley transformation 
    $Q =\text{Cayley}(J)= (I+Z)(I-Z)^{-1}$
where $Z = J^\top-J$, so that $Q^\top Q=I$ for any $J\in\mathbb{R}^{m+p\times m+p}$. This parameterization can be applied for unconstrained learning together with the direct parametrization of biLipREN.
\end{remark}

The composition of an all-pass filter and a stable invertible nonlinear system is a form of inner-outer factorization of a nonlinear system \cite{ball1992inner}, generalising the classical minimum-phase/all-pass factorization of linear systems. The input-output map of a nonlinear system $\bm{M}$ can be written as a composition of two operators $\bm{M}=\bm{\Sigma}\circ\bm{\Theta}$ where $\bm{\Theta}$ is a stable invertible nonlinear system and $\bm{\Sigma}$ is distance preserving, i.e. the distance between any two signals remains unchanged. Related forms of factorization have proven useful in many control designs, e.g., Smith predictor \cite{laughlin1987smith}, $H_\infty $ control \cite{ball1996j} and spectral factorization \cite{oara1998solutions}. 

\begin{remark}
In our proposed framework, the all-pass/inner factor is linear (in fact affine due to the bias terms), whereas in previous approaches to nonlinear inner-outer factorization the inner factor is energy-preserving ($\|y\|=\|u\|$) but may be nonlinear (e.g. \cite{ball1992inner, van2018l2}). It is reasonable to ask whether in the incremental setting the inner factor could also be nonlinear. However, the Mazur–Ulam theorem \cite{mazur1932transformations, nica2012mazur} states that on any real normed linear space (even infinite dimensional), every invertible distance-preserving map is affine. So in this sense, there is no loss of generality from restricting to affine models.
\end{remark}

\section{Bi-Lipschitz Recurrent Equilibrium Networks}
In this section we present a parameterization of recurrent neural models with guaranteed robust invertibility. More specifically, the neural network dynamical model takes the form of \emph{Recurrent Equilibrium Network} (REN) \cite{revay2024recurrent}. We will give conditions for a REN to be contracting and bi-Lipschitz. We then show that its inverse admits a REN realization which is also contracting and bi-Lipschitz.

\subsection{Recurrent Equilibrium Networks}
We parameterize \eqref{eq:system} in the form of recurrent equilibrium network (REN) \cite{revay2024recurrent}, which is a feedback interconnection of a linear time-invariant system and a static nonlinearity:
\begin{subequations}\label{eq:ren}
\begin{align}
\renewcommand\arraystretch{1.2}
\begin{bmatrix}
x_{t+1} \\ v_t \\ y_t
\end{bmatrix}&=
\overset{W}{\overbrace{
		\left[
		\begin{array}{c|cc}
		A & B_1 & B_2 \\ \hline 
		C_{1} & D_{11} & D_{12} \\
		C_{2} & D_{21} & D_{22}
		\end{array} 
		\right]
}}
\begin{bmatrix}
x_t \\ w_t \\ u_t
\end{bmatrix}+
\overset{b}{\overbrace{
		\begin{bmatrix}
		b_x \\ b_v \\ b_y
		\end{bmatrix}
}}, \label{eq:G}\\
w_t=\sigma(&v_t):=\begin{bmatrix}
\sigma(v_{t}^1) & \sigma(v_{t}^2) & \cdots & \sigma(v_{t}^q)
\end{bmatrix}^\top, \label{eq:sigma}
\end{align}   
\end{subequations}
where $x_t\in \R^n, u_t,y_t\in \R^m$ are the state, input and output, respectively. $v_t, w_t\in \R^q$ are the input-output of the activation layer $\sigma$, which is a fixed nonlinear scalar activation function with slope-restricted in $[0,1]$. The learnable parameter $\theta$ includes weight matrix $W$ and bias vector $b$. 
The algebraic loop in \eqref{eq:ren} introduces an \emph{equilibrium network} $\phi:(x_t,u_t)\rightarrow w_t$, where $w_t$ is the solution of the following implicit equation
\begin{equation}\label{eq:implicit}
w_t=\sigma(D_{11} w_t+b_w),
\end{equation}
with $b_w=C_1x_t+D_{12}u_t+b_v$. Thus, REN \eqref{eq:ren} can be rewritten as \eqref{eq:system} with 
\[
\begin{split}
    f_\theta(x,u)&=Ax+B_1\phi(x,u)+B_2u+b_x, \\
    h_\theta(x,u)&=C_2x+D_{21}\phi(x,u) + D_{22}u+b_y.
\end{split}
\]

One advantage of the equilibrium network \eqref{eq:implicit} is that it includes many existing feed-forward networks| e.g., multi-layer perception (MLP), residual network (ResNet) and convolutional neural network (CNN)| as special cases \cite{ghaoui2019implicit}. 

Due to the implicit structure in \eqref{eq:implicit}, we need pay additional attention to the \emph{well-posedness} of $\phi$. That is, for any $b_w\in \mathbb{R}^q$, Eq.~\eqref{eq:implicit} admits a unique solution $w_t\in \R^q$. In \cite{revay2020lipschitz} it was shown that if there exists a $\Lambda\in \mathbb{D}^+$ with $\mathbb{D}_+$ as the set of positive diagonal matrices such that 
\begin{equation}\label{eq:wellpose}
    2\Lambda-\Lambda  D_{11}- D_{11}^\top\Lambda\succ0,
\end{equation}
then the equilibrium network \eqref{eq:implicit} is well-posed.

\subsection{BiLipREN and Its Contracting Inverse}
We introduce the explicit model of REN inverse by taking $u_t=D_{22}^{-1}(-C_2x_t-D_{21}w_t+y-b_y)$ and substitute it back to \eqref{eq:ren}. This inversion is well-defined if $D_{22}$ is non-singular. Then, the REN inverse can also be written as a REN:
\begin{equation}\label{eq:iren}
    \begin{split}
    \renewcommand\arraystretch{1.2}
\begin{bmatrix}
x_{t+1} \\ v_t \\ u_t
\end{bmatrix}&=
\overset{\hat W}{\overbrace{
		\left[
		\begin{array}{c|cc}
		\hat A & \hat B_1 & \hat B_2 \\ \hline 
		\hat C_{1} & \hat D_{11} & \hat D_{12} \\
		\hat C_{2} & \hat D_{21} & \hat D_{22}
		\end{array} 
		\right]
}}
\begin{bmatrix}
x_t \\ w_t \\ y_t
\end{bmatrix}+
\overset{\hat b}{\overbrace{
		\begin{bmatrix}
		\hat b_x \\ \hat b_v \\ \hat b_y
		\end{bmatrix}
}} \\
        w_t&=\sigma(v_t)
    \end{split}
\end{equation}
where the explicit parameter of REN inverse is 
\begin{equation}\label{eq:irenparam}
    \begin{split}
        &\hat{A}=A-B_2D_{22}^{-1}C_2,\; \hat{B}_1 = B_1-B_2D_{22}^{-1}D_{21}, \\
        &\hat{B}_2 = B_2D_{22}^{-1},\; \hat{C}_1 = C_1-D_{12}D_{22}^{-1}C_2\\
        & \hat{C}_2=-D_{22}^{-1}C_2,\;
        \hat{D}_{11}=D_{11}-D_{12}D_{22}^{-1}D_{21},\\
        &\hat{D}_{12}=D_{12}D_{22}^{-1},\;\hat{D}_{21}=-D_{22}^{-1}D_{21},\; \hat{D}_{22}=D_{22}^{-1}\\
        &\hat{b}_x = b_x-B_2D_{22}^{-1}b_y,\; \hat{b}_v = b_v-D_{12}D_{22}^{-1}b_y,\\
        &\hat{b}_y=-D_{22}^{-1}b_y.
    \end{split}
\end{equation}
It is easy to verify that the parametrization of REN inverse is a coordinate transformation $\hat{W}= W\Psi$ and $\hat{b}= b\Psi$ with 
\begin{equation}\label{eq:Psi}
    \Psi=\begin{bmatrix}
        I & 0 & 0 \\
        0 & I & 0 \\
        -D_{22}^{-1}C_2 & -D_{22}^{-1}D_{21} & D_{22}^{-1}
    \end{bmatrix}.
\end{equation}
We now state our main theoretical result as follows. 
\begin{thm}\label{thm:main}
Suppose that there exist $P = P^\top\succ 0$ and $\Lambda \in \mathbb{D}^+$ satisfying the following matrix inequality
\begin{align}
   &\begin{bmatrix}\nonumber
       P& -C_1^\top \Lambda&C_2^\top \\
       -\Lambda C_1 & \mathcal{W} &D_{21}^\top - \Lambda D_{12}\\
       C_2&D_{21}- D_{12}^\top\Lambda &-\frac{2}{\mu+\nu} I+D_{22}+D_{22}^\top.
   \end{bmatrix} \\
   &\; -
   \begin{bmatrix}
       A^\top\\
       B_1^\top\\
       B_2^\top 
   \end{bmatrix}P\begin{bmatrix}
       A^\top\\
       B_1^\top\\
       B_2^\top 
   \end{bmatrix}^\top-\frac{2\mu\nu}{\mu+\nu}\begin{bmatrix}
       C_2^\top\\
       D_{21}^\top\\
       D_{22}^\top
   \end{bmatrix}\begin{bmatrix}
       C_2^\top\\
       D_{21}^\top\\
       D_{22}^\top
       \end{bmatrix}^\top\succ 0 \label{eq:passivity}
    \end{align}
where $\mathcal{W}=2\Lambda -\Lambda D_{11}-D_{11}^\top \Lambda$.
\begin{itemize}
    \item[1)] REN \eqref{eq:ren} is well-posed, contracting and $(\mu,\nu)$-biLipschitz.
    \item[2)] The REN inverse \eqref{eq:iren} is well-posed, contracting and $(1/\nu,1/\mu)$-biLipschitz.  
\end{itemize}
\end{thm}

\begin{remark}
    By applying the model direct parameterization method in \cite{revay2024recurrent} with particular choice of supply rate matrices using \eqref{eq:monotone}, we can construct a mapping $\mathcal{M}:\theta \rightarrow (W,b)$ such that Condition~\eqref{eq:passivity} holds for any $\theta \in \R^N$.
\end{remark}

\begin{proof}
For Claim 1), the well-posedness of \eqref{eq:ren} can be directly drawn from \eqref{eq:passivity} as $\mathcal{W}\succ 0$, i.e., \eqref{eq:wellpose} holds. Then, direct application of \cite[Thm.~1]{revay2024recurrent} gives us that \eqref{eq:ren} is contracting and satisfies the strong monotonicity in Definition \ref{dfn:monotone}. By Lemma~\ref{lem:iqc}, it is also $(\mu,\nu)$-biLipschitz.

We now begin to prove Claim 2). First, we show that $D_{22}$ is invertible since \eqref{eq:passivity} implies 
\[
D_{22}+D_{22}^\top \succeq \frac{2}{\mu+\nu}I.
\]
Thus, the matrix $\Psi$ in \eqref{eq:Psi} and $(\hat{W},\hat{b})$ are well-defined. By left- and right-multiplying \eqref{eq:passivity} with $\Psi^\top$ and $\Psi$, respectively, we can obtain
\begin{align}
   &\begin{bmatrix}\nonumber
       P& -\hat C_1^\top \Lambda & \hat C_2^\top \\
       -\Lambda \hat C_1 & \hat{\mathcal{W}} & \hat D_{21}^\top - \Lambda \hat D_{12}\\
       \hat C_2 & \hat D_{21}- \hat D_{12}^\top\Lambda &-\frac{2}{\beta+\gamma} I+\hat D_{22}+\hat D_{22}^\top.
   \end{bmatrix} \\
   &\; -
   \begin{bmatrix}
       \hat A^\top\\
       \hat B_1^\top\\
       \hat B_2^\top 
   \end{bmatrix}P\begin{bmatrix}
       \hat A^\top\\
       \hat B_1^\top\\
       \hat B_2^\top 
   \end{bmatrix}^\top-\frac{2\beta\gamma}{\beta+\gamma}\begin{bmatrix}
       \hat C_2^\top\\
       \hat D_{21}^\top\\
       \hat D_{22}^\top
   \end{bmatrix}\begin{bmatrix}
       \hat C_2^\top\\
       \hat D_{21}^\top\\
       \hat D_{22}^\top
       \end{bmatrix}^\top\succ 0 
\end{align}
where $\beta=1/\nu$, $\gamma=1/\mu$ and $\hat{\mathcal{W}}=2\Lambda -\Lambda \hat D_{11}-\hat D_{11}^\top \Lambda$. Again, by \cite[Thm.~1]{revay2024recurrent} we have that the REN inverse \eqref{eq:iren} is well-posed, contracting and $(\beta,\gamma)$-biLipschitz.
\end{proof}

\section{Experiments}
In this section we present experiments which illustrate the utility of our proposed models, system identification via inner-outer factorization, and robust inversion of an minimum phase nonlinear system.
\subsection{Nonlinear Inner-outer Factorization} 
In the first example, we consider learning the inner-outer factorization of a nonlinear system with time delay $\dot x(t) = 0.9\text{tanh}x(t)+u(t-1)$, which is a stable nonlinear system coupled with a time delay.
The inner-outer factorization is a composition of an all-pass filter and a stable invertible system which naturally fits our proposed approach with a composition of a dynamics orthogonal layer and a biLipREN as the inner and outer system respectively. We plot the open-loop response of the outer system and the original system under the Gaussian noise input to demonstrate the learned factorization in Figure \ref{fig:fact}. 

We compare the system response between the outer system, the composed system and the original system. The learned response of the composed system fits the original system, with a time delay of $1$s. The outer system, or the biLipREN recovers the system response of the minimum phase without the time delay. We also show the impulse response of the inner system, while the dynamic orthogonal layer models the $1$s time delay and the magnitude of the inner system is $1$, which does not amplify the outer system. 

The data is generated by simulating the system for $15$s and sampling $100$ batch data points with $100$ time steps under a Gaussian noise input. We fit the $(0.1,5)$-biLipschitz dynamic models to the data by minimizing the fitting error in $\ell_2$-norm $J = \|\bm{F}(u_k)-y_k\|^2$, where $u_k$ and $y_k$ are the input and output of the $k$th batch data respectively. BiLipREN has $3$ internal states and $30$ neurons, and the dynamic orthogonal layer has $30$ internal states. sequences. 
\begin{figure}[!bt]
\centering
    \includegraphics[width=\linewidth]{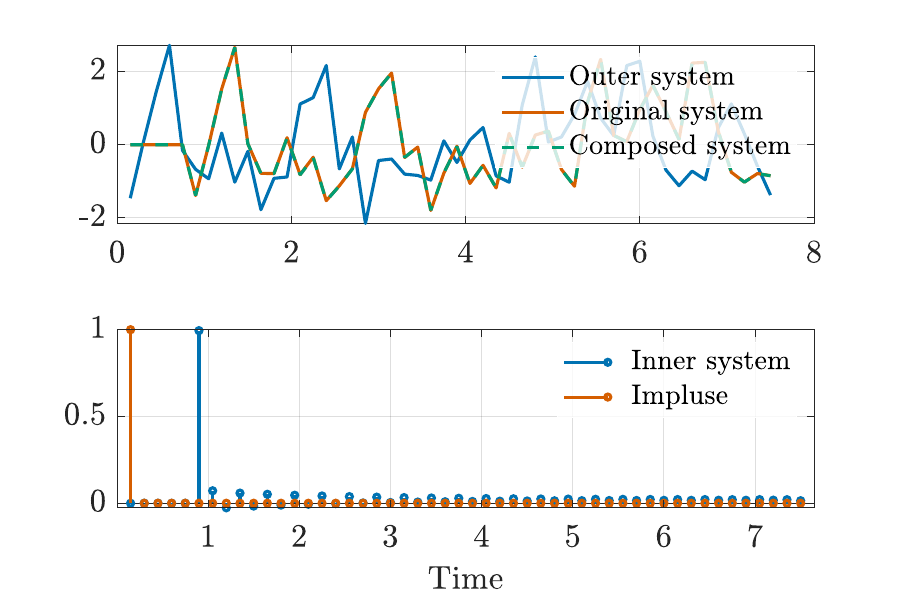}
    \caption{Open loop simulation of the outer system, the original system and the composed system under the Gaussian noise (Top), the impulse response of the inner system (Bottom).}\label{fig:fact} 
\end{figure}
\subsection{Robust Inversion}
In the second experiment, we learn a robustly-invertible model of a nonlinear mechanical system with four coupled mass spring dampers with nonlinear spring characteristic, adapted from \cite{revay2020convex}. A schematic is shown in Figure \ref{fig:msd}. Since the linearized system is easily verified to be minimum phase phase, we attempt to learn a nonlinear invertible model, such that the system input can be robustly reconstructed from the output, even with the addition of measurement noise and uncertain initial conditions.

The nonlinearity in the system introduced through the spring's piece-wise linear force profile
\[
F_i(d) = k_i \Gamma(d), \quad \Gamma(d)=
\begin{cases}
d + 0.75, & d \leq -1, \\
0.25\,d, & -1 < d < 1, \\
d - 0.75, & d \geq 1.
\end{cases}
\]
where $k_i$ is the spring constant for the $i$th spring and $d$ is the displacement between the carts. We excite the first cart with a piece-wise constant input signal that changes value after an interval distributed uniformly in $[0,\tau]$, and takes values that are normally distributed with standard deviation $\sigma$. We measure the position of the last cart, and measurement noise is modeled as zero-mean Gaussian noise with a signal-to-noise ratio of approximately $30$dB. The system is simulated for $20$s and sampled at $50\mathrm{Hz}$ to generate $1000$ data points with an input signal characterized by $\tau = 20$s and $\sigma = 3$N. 

We use a 4 layer $(0.1,5)$-biLipschitz model, and each layer is composed by a biLipREN with $50$ internal states and $60$ neurons and a static orthogonal layer to fit the system. We compare the train and test error  to a contracting REN with same number of internal states and neurons. The fitting result from one trajectories is shown in Figure \ref{fig:fitting}. Notice that the training error and test error using the biLipREN is larger than the contracting REN, as shown in the Figure \ref{fig:error}. This is to be expected, since the biLipREN represents a restricted model class relative to the contracting REN, but the effect is mild. The zoomed in section shows the fitting details, the contracting REN smoothly fits the output sequence while the biLipREN has a slight shift. The fitting error is evaluated by normalized simulation error (NSE) between the training data and the output of the network \begin{equation}\nonumber
    \mathrm{NSE}= \frac{\|\G(u_k)-y_k\|}{\|y_k\|}
\end{equation}
where $u_k$ and $y_k$ are the input and output of the $k$th batch data respectively. However, biLipREN can robustly recover the input sequence with the presence of noise in the output and difference in initial state. In contrast, the contracting REN does not have any guarantee that an inverse exists (i.e. is well-posed), if it exists that it is stable or Lipschitz. One of the reconstructed input sequences from biLipREN is shown in Figure \ref{fig:inverse}. 

 We compare reconstruction error to the theoretical bound calculated as in \eqref{eq:error-bound}. We use adversarial training to find the worst case of the actual error, by setting the input sequence $u$, perturbed input sequence $u_p$ and the initial condition $b$ of the inverse process as training parameters. We constrain the norm of the difference in initial value within $0.1$ and the norm of perturbation within $1$ by projected gradient descent. For the calculation of theoretical bound of our trained model, the maximum contracting rate $\bar \alpha =0.9$, the overshoot $\kappa=\sqrt{\frac{\bar \sigma}{\underline{\sigma}}}\approx3.03$, where $\bar \sigma$ and $\underline{\sigma}$ are the maximum and minimum singular value of $P$ in \eqref{eq:passivity} respectively, see \cite{revay2024recurrent}. The bi-Lipschitz bound of the output layer $h_\theta$ is found by calculating the maximum and minimum singular value of the Jacobian matrix with respect to the states and inputs, which is $(0.71,6.54)$. The time-averaged reconstruction error:
 \[
\frac{1}{T} \|\G_b^{-1}(\G_a(u+\delta_u)-u\|_T
 \]is plotted as a function of $T$ in Figure \ref{fig:bound} along with its bound \eqref{eq:error-bound}. It can be seen that the theoretical prediction is verified.

\begin{figure}[!bt]
\centering
    \includegraphics[width=\linewidth]{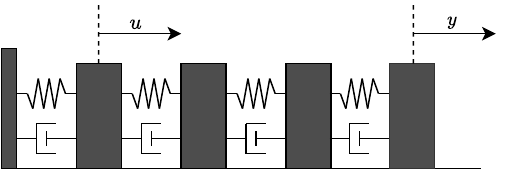}
    \caption{Four coupled nonlinear mass spring dampers.}\label{fig:msd} 
\end{figure}
\begin{figure}[!bt]
\centering
    \includegraphics[width=\linewidth]{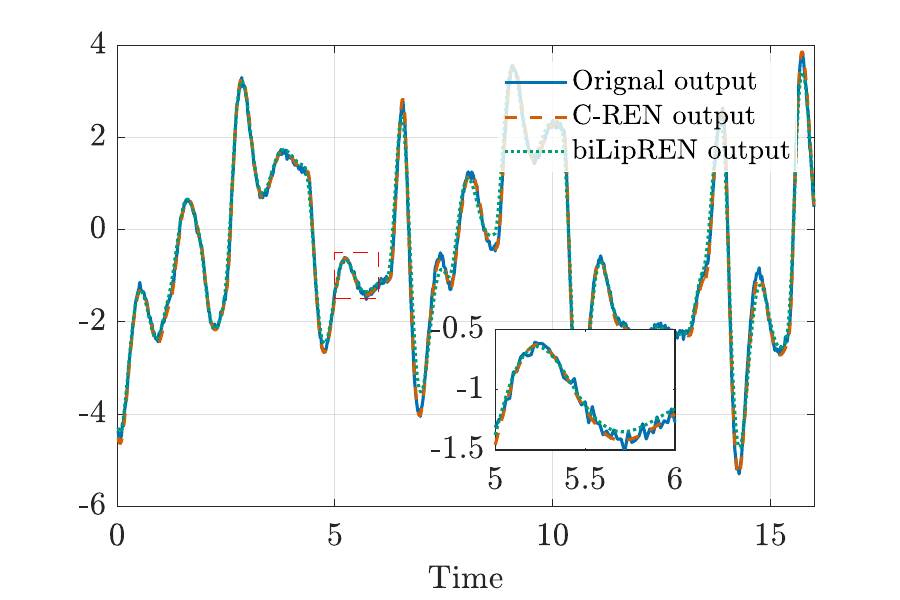}
    \caption{One of output sequences using a biLipREN and a contracting REN to fit a 4 coupled mass spring damper.}\label{fig:fitting} 
\end{figure}
\begin{figure}[!bt]
\centering
    \includegraphics[width=\linewidth]{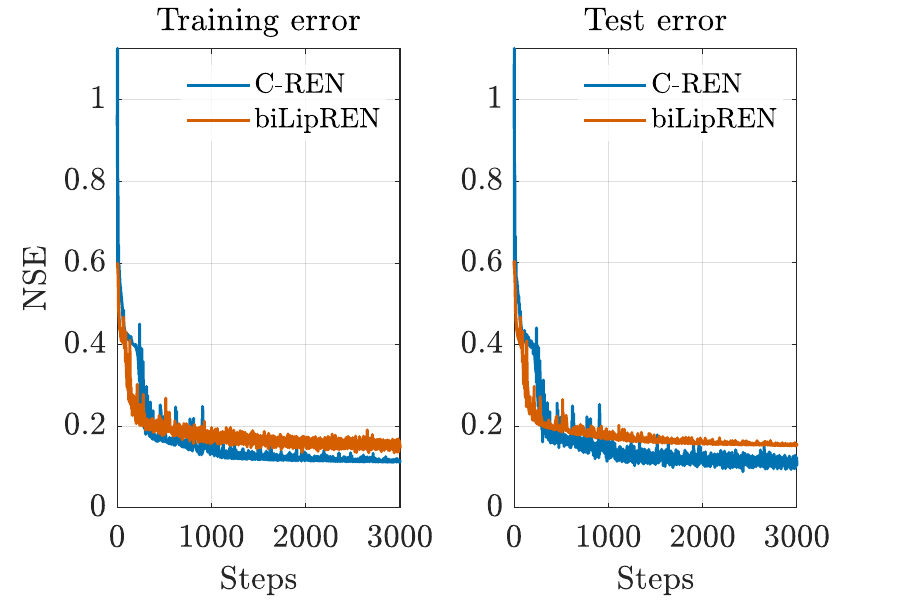}
    \caption{Normalized training error (left) and test error (right) of biLipREN and contracting REN.}\label{fig:error} 
\end{figure}
\begin{figure}[!bt]
\centering
    \includegraphics[width=\linewidth]{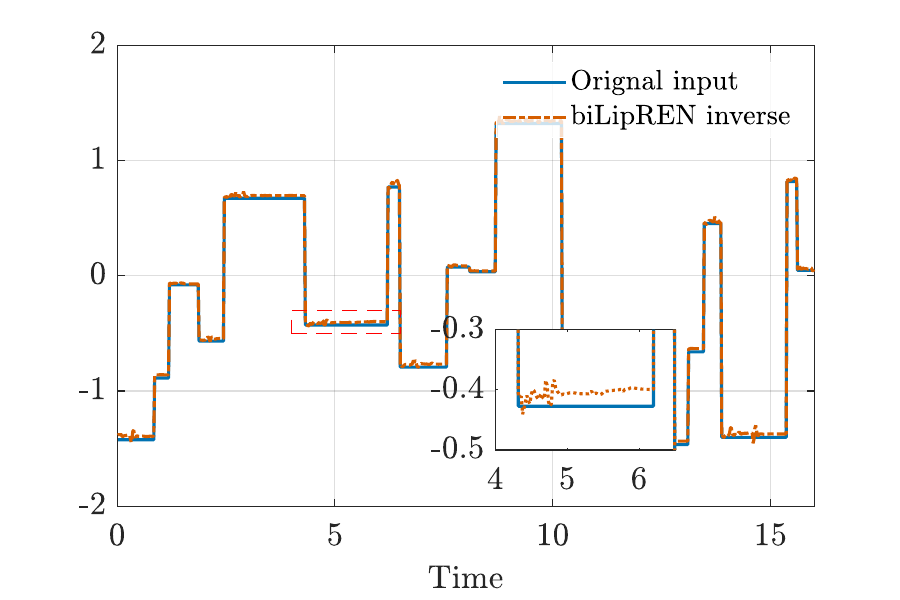}
    \caption{Reconstructed input sequence from output sequence in Figure \ref{fig:fitting} by biLipREN with measurement noise .}\label{fig:inverse} 
\end{figure}
\begin{figure}[!bt]
\centering
    \includegraphics[width=\linewidth]{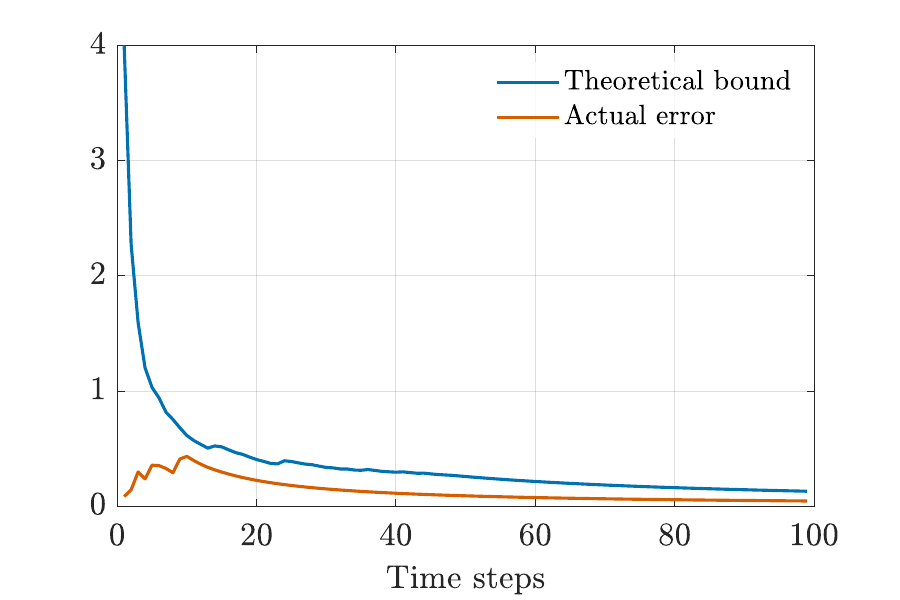}
    \caption{Time average of the actual error and theoretical bound}\label{fig:bound} 
\end{figure}

\section{Conclusions}
In this paper, we have formally defined robust invertibility for nonlinear systems, in terms of contraction and bi-Lipschitzness. We use strong input-output monotonicity to describe a bi-Lipschitz operator. By composing with static orthogonal layers, we obtain more general bi-Lipschitz dynamic models with guaranteed robust invertibility. By composing with dynamic orthogonal layers, we obtain a bi-Lipschitz dynamic models with a non-causal inverse.

Building on the notion of robust invertibility, we construct a bi-Lipschitz recurrent equilibrium network (biLipREN), which has a robust inverse and direct parameterization enabling learning via unconstrained optimization. We illustrate the approach with a nonlinear inner-outer factorization and a robust inversion experiment.

%\addtolength{\textheight}{-12cm}   % This command serves to balance the column lengths
                                  % on the last page of the document manually. It shortens
                                  % the textheight of the last page by a suitable amount.
                                  % This command does not take effect until the next page
                                  % so it should come on the page before the last. Make
                                  % sure that you do not shorten the textheight too much.

%%%%%%%%%%%%%%%%%%%%%%%%%%%%%%%%%%%%%%%%%%%%%%%%%%%%%%%%%%%%%%%%%%%%%%%%%%%%%%%%

%%%%%%%%%%%%%%%%%%%%%%%%%%%%%%%%%%%%%%%%%%%%%%%%%%%%%%%%%%%%%%%%%%%%%%%%%%%%%%%%

%%%%%%%%%%%%%%%%%%%%%%%%%%%%%%%%%%%%%%%%%%%%%%%%%%%%%%%%%%%%%%%%%%%%%%%%%%%%%%%%
% \section*{APPENDIX}

% \section*{ACKNOWLEDGMENT}

%%%%%%%%%%%%%%%%%%%%%%%%%%%%%%%%%%%%%%%%%%%%%%%%%%%%%%%%%%%%%%%%%%%%%%%%%%%%%%%%

% \balance
\bibliographystyle{IEEEtran}
\bibliography{ref.bib}

\begin{thebibliography}{10}
\providecommand{\url}[1]{#1}
\csname url@rmstyle\endcsname
\providecommand{\newblock}{\relax}
\providecommand{\bibinfo}[2]{#2}
\providecommand\BIBentrySTDinterwordspacing{\spaceskip=0pt\relax}
\providecommand\BIBentryALTinterwordstretchfactor{4}
\providecommand\BIBentryALTinterwordspacing{\spaceskip=\fontdimen2\font plus
\BIBentryALTinterwordstretchfactor\fontdimen3\font minus \fontdimen4\font\relax}
\providecommand\BIBforeignlanguage[2]{{%
\expandafter\ifx\csname l@#1\endcsname\relax
\typeout{** WARNING: IEEEtran.bst: No hyphenation pattern has been}%
\typeout{** loaded for the language `#1'. Using the pattern for}%
\typeout{** the default language instead.}%
\else
\language=\csname l@#1\endcsname
\fi
#2}}

\bibitem{zames1981feedback}
G.~Zames, ``Feedback and optimal sensitivity: Model reference transformations, multiplicative seminorms, and approximate inverses,'' \emph{IEEE Transactions on automatic control}, vol.~26, no.~2, pp. 301--320, 1981.

\bibitem{markusson2001iterative}
O.~Markusson, H.~Hjalmarsson, and M.~Norrlof, ``Iterative learning control of nonlinear non-minimum phase systems and its application to system and model inversion,'' in \emph{Proceedings of the 40th IEEE Conference on Decision and Control}, 2001.

\bibitem{rigney2009nonminimum}
B.~P. Rigney, L.~Y. Pao, and D.~A. Lawrence, ``Nonminimum phase dynamic inversion for settle time applications,'' \emph{IEEE Transactions on Control Systems Technology}, vol.~17, no.~5, pp. 989--1005, 2009.

\bibitem{van2018inversion}
J.~van Zundert and T.~Oomen, ``On inversion-based approaches for feedforward and ilc,'' \emph{Mechatronics}, vol.~50, pp. 282--291, 2018.

\bibitem{bolderman2024physics}
M.~Bolderman, H.~Butler, S.~Koekebakker, E.~Van~Horssen, R.~Kamidi, T.~Spaan-Burke, N.~Strijbosch, and M.~Lazar, ``Physics-guided neural networks for feedforward control with input-to-state-stability guarantees,'' \emph{Control Engineering Practice}, vol. 145, p. 105851, 2024.

\bibitem{morgan2006generalized}
D.~R. Morgan, Z.~Ma, J.~Kim, M.~G. Zierdt, and J.~Pastalan, ``A generalized memory polynomial model for digital predistortion of {RF} power amplifiers,'' \emph{IEEE Transactions on signal processing}, vol.~54, no.~10, pp. 3852--3860, 2006.

\bibitem{ghannouchi2009behavioral}
F.~M. Ghannouchi and O.~Hammi, ``Behavioral modeling and predistortion,'' \emph{IEEE Microwave magazine}, vol.~10, no.~7, pp. 52--64, 2009.

\bibitem{liu2006augmented}
T.~Liu, S.~Boumaiza, and F.~M. Ghannouchi, ``Augmented hammerstein predistorter for linearization of broad-band wireless transmitters,'' \emph{IEEE transactions on microwave theory and techniques}, vol.~54, no.~4, pp. 1340--1349, 2006.

\bibitem{tanovic2018equivalent}
O.~Tanovic, A.~Megretski, Y.~Li, V.~Stojanovic, and M.~Osqui, ``Equivalent baseband models and corresponding digital predistortion for compensating dynamic passband nonlinearities in phase-amplitude modulation-demodulation schemes,'' \emph{IEEE Transactions on Signal Processing}, vol.~66, no.~22, pp. 5972--5987, 2018.

\bibitem{reiner1996flight}
J.~Reiner, G.~J. Balas, and W.~L. Garrard, ``Flight control design using robust dynamic inversion and time-scale separation,'' \emph{Automatica}, vol.~32, no.~11, pp. 1493--1504, 1996.

\bibitem{miller2011nonlinear}
C.~Miller, ``Nonlinear dynamic inversion baseline control law: architecture and performance predictions,'' in \emph{AIAA Guidance, Navigation, and Control Conference}, 2011, p. 6467.

\bibitem{sieberling2010robust}
S.~Sieberling, Q.~Chu, and J.~Mulder, ``Robust flight control using incremental nonlinear dynamic inversion and angular acceleration prediction,'' \emph{Journal of guidance, control, and dynamics}, vol.~33, no.~6, pp. 1732--1742, 2010.

\bibitem{wang2017robust}
L.~Wang, A.~Isidori, Z.~Liu, and H.~Su, ``Robust output regulation for invertible nonlinear mimo systems,'' \emph{Automatica}, vol.~82, pp. 278--286, 2017.

\bibitem{garcia1982internal}
C.~E. Garcia and M.~Morari, ``Internal model control. a unifying review and some new results,'' \emph{Industrial \& Engineering Chemistry Process Design and Development}, vol.~21, no.~2, pp. 308--323, 1982.

\bibitem{economou1986internal}
C.~G. Economou, M.~Morari, and B.~O. Palsson, ``Internal model control: Extension to nonlinear system,'' \emph{Industrial \& Engineering Chemistry Process Design and Development}, vol.~25, no.~2, pp. 403--411, 1986.

\bibitem{celani2010output}
F.~Celani and A.~Isidori, ``Output stabilization of strongly minimum-phase systems,'' \emph{IFAC Proceedings Volumes}, vol.~43, no.~14, pp. 647--652, 2010.

\bibitem{lohmiller1998contraction}
W.~Lohmiller and J.-J.~E. Slotine, ``On contraction analysis for non-linear systems,'' \emph{Automatica}, vol.~34, no.~6, pp. 683--696, 1998.

\bibitem{FB-CTDS}
\BIBentryALTinterwordspacing
F.~Bullo, \emph{Contraction Theory for Dynamical Systems}, {1.2}~ed.\hskip 1em plus 0.5em minus 0.4em\relax Kindle Direct Publishing, 2024. [Online]. Available: \url{https://fbullo.github.io/ctds}
\BIBentrySTDinterwordspacing

\bibitem{moylan1977stable}
P.~Moylan, ``Stable inversion of linear systems,'' \emph{IEEE Transactions on Automatic Control}, vol.~22, no.~1, pp. 74--78, 1977.

\bibitem{byrnes1988local}
C.~I. Byrnes and A.~Isidori, ``Local stabilization of minimum-phase nonlinear systems,'' \emph{Systems \& Control Letters}, vol.~11, no.~1, pp. 9--17, 1988.

\bibitem{ball1992inner}
J.~A. Ball and J.~W. Helton, ``Inner-outer factorization of nonlinear operators,'' \emph{Journal of functional analysis}, vol. 104, no.~2, pp. 363--413, 1992.

\bibitem{van2018l2}
A.~van~der Schaft, \emph{L2-Gain and Passivity Techniques in Nonlinear Control}, 2nd~ed.\hskip 1em plus 0.5em minus 0.4em\relax Springer Publishing Company, Incorporated, 2018.

\bibitem{zou2007precision}
Q.~Zou and S.~Devasia, ``Precision preview-based stable-inversion for nonlinear nonminimum-phase systems: The vtol example,'' \emph{Automatica}, vol.~43, no.~1, pp. 117--127, 2007.

\bibitem{romagnoli2019general}
R.~Romagnoli and E.~Garone, ``A general framework for approximated model stable inversion,'' \emph{Automatica}, vol. 101, pp. 182--189, 2019.

\bibitem{laughlin1987smith}
D.~L. Laughlin, D.~E. Rivera, and M.~Morari, ``Smith predictor design for robust performance,'' \emph{International Journal of Control}, vol.~46, no.~2, pp. 477--504, 1987.

\bibitem{hirschorn1979invertibility}
R.~Hirschorn, ``Invertibility of nonlinear control systems,'' \emph{SIAM Journal on Control and Optimization}, vol.~17, no.~2, pp. 289--297, 1979.

\bibitem{devasia1996nonlinear}
S.~Devasia, D.~Chen, and B.~Paden, ``Nonlinear inversion-based output tracking,'' \emph{IEEE Transactions on Automatic Control}, vol.~41, no.~7, pp. 930--942, 1996.

\bibitem{liberzon2002output}
D.~Liberzon, A.~S. Morse, and E.~D. Sontag, ``Output-input stability and minimum-phase nonlinear systems,'' \emph{IEEE Transactions on Automatic Control}, vol.~47, no.~3, pp. 422--436, 2002.

\bibitem{liberzon2004output}
D.~Liberzon, ``Output--input stability implies feedback stabilization,'' \emph{Systems \& control letters}, vol.~53, no. 3-4, pp. 237--248, 2004.

\bibitem{papamakarios2021normalizing}
G.~Papamakarios, E.~Nalisnick, D.~J. Rezende, S.~Mohamed, and B.~Lakshminarayanan, ``Normalizing flows for probabilistic modeling and inference,'' \emph{Journal of Machine Learning Research}, vol.~22, no.~57, pp. 1--64, 2021.

\bibitem{luo2024potential}
Y.~Luo, C.~Sun, J.~B. Tenenbaum, and Y.~Du, ``Potential based diffusion motion planning,'' in \emph{International Conference on Machine Learning}.\hskip 1em plus 0.5em minus 0.4em\relax PMLR, 2024, pp. 33\,486--33\,510.

\bibitem{urain2020imitationflow}
J.~Urain, M.~Ginesi, D.~Tateo, and J.~Peters, ``Imitationflow: Learning deep stable stochastic dynamic systems by normalizing flows,'' in \emph{2020 IEEE/RSJ International Conference on Intelligent Robots and Systems (IROS)}.\hskip 1em plus 0.5em minus 0.4em\relax IEEE, 2020, pp. 5231--5237.

\bibitem{de2020normalizing}
E.~de~B{\'e}zenac, S.~S. Rangapuram, K.~Benidis, M.~Bohlke-Schneider, R.~Kurle, L.~Stella, H.~Hasson, P.~Gallinari, and T.~Januschowski, ``Normalizing kalman filters for multivariate time series analysis,'' \emph{Advances in Neural Information Processing Systems}, vol.~33, pp. 2995--3007, 2020.

\bibitem{drygala2022generative}
C.~Drygala, B.~Winhart, F.~di~Mare, and H.~Gottschalk, ``Generative modeling of turbulence,'' \emph{Physics of Fluids}, vol.~34, no.~3, 2022.

\bibitem{chen2019residual}
R.~T. Chen, J.~Behrmann, D.~K. Duvenaud, and J.-H. Jacobsen, ``Residual flows for invertible generative modeling,'' \emph{Advances in Neural Information Processing Systems}, vol.~32, 2019.

\bibitem{behrmann2021understanding}
J.~Behrmann, P.~Vicol, K.-C. Wang, R.~Grosse, and J.-H. Jacobsen, ``Understanding and mitigating exploding inverses in invertible neural networks,'' in \emph{International Conference on Artificial Intelligence and Statistics}.\hskip 1em plus 0.5em minus 0.4em\relax PMLR, 2021, pp. 1792--1800.

\bibitem{wangmonotone}
R.~Wang, K.~D. Dvijotham, and I.~R. Manchester, ``Monotone, bi-lipschitz, and polyak-lojasiewicz networks,'' in \emph{Forty-first International Conference on Machine Learning}, 2024.

\bibitem{revay2024recurrent}
M.~Revay, R.~Wang, and I.~R. Manchester, ``Recurrent equilibrium networks: Flexible dynamic models with guaranteed stability and robustness,'' \emph{IEEE Transactions on Automatic Control}, vol.~69, no.~5, pp. 2855--2870, 2024.

\bibitem{chaffey2023graphical}
T.~Chaffey, F.~Forni, and R.~Sepulchre, ``Graphical nonlinear system analysis,'' \emph{IEEE Transactions on Automatic Control}, vol.~68, no.~10, pp. 6067--6081, 2023.

\bibitem{trockman2021orthogonalizing}
A.~Trockman and J.~Z. Kolter, ``Orthogonalizing convolutional layers with the cayley transform,'' \emph{arXiv preprint arXiv:2104.07167}, 2021.

\bibitem{heuberger2005modelling}
P.~S. Heuberger, P.~M. van~den Hof, and B.~Wahlberg, \emph{Modelling and Identification with Rational Orthogonal Basis Functions}.\hskip 1em plus 0.5em minus 0.4em\relax Springer Science \& Business Media, 2005.

\bibitem{ball1996j}
J.~A. Ball and A.~J. Van~der Schaft, ``J-inner-outer factorization, j-spectral factorization, and robust control for nonlinear systems,'' \emph{IEEE Transactions on Automatic Control}, vol.~41, no.~3, pp. 379--392, 1996.

\bibitem{oara1998solutions}
C.~Oara and A.~Varga, ``Solutions to the general inner-outer and spectral factorization problems,'' in \emph{Proceedings of the 37th IEEE Conference on Decision and Control (Cat. No. 98CH36171)}, vol.~3.\hskip 1em plus 0.5em minus 0.4em\relax IEEE, 1998, pp. 2774--2779.

\bibitem{mazur1932transformations}
S.~Mazur and S.~Ulam, ``Sur les transformations isom{\'e}triques d’espaces vectoriels norm{\'e}s,'' \emph{CR Acad. Sci. Paris}, vol. 194, no. 946-948, p. 116, 1932.

\bibitem{nica2012mazur}
B.~Nica, ``The mazur--ulam theorem,'' \emph{Expositiones Mathematicae}, vol.~30, no.~4, pp. 397--398, 2012.

\bibitem{ghaoui2019implicit}
L.~El~Ghaoui, F.~Gu, B.~Travacca, A.~Askari, and A.~Tsai, ``Implicit deep learning,'' \emph{SIAM Journal on Mathematics of Data Science}, vol.~3, no.~3, pp. 930--958, 2021.

\bibitem{revay2020lipschitz}
M.~Revay, R.~Wang, and I.~R. Manchester, ``Lipschitz bounded equilibrium networks,'' \emph{arXiv preprint arXiv:2010.01732}, 2020.

\bibitem{revay2020convex}
------, ``A convex parameterization of robust recurrent neural networks,'' \emph{IEEE Control Systems Letters}, vol.~5, no.~4, pp. 1363--1368, 2020.

\end{thebibliography}
\end{document}